\documentclass{llncs}
\usepackage{graphicx}
\usepackage{multirow}
\usepackage{bigstrut}
\usepackage{amssymb}
\usepackage{amsfonts}
\usepackage{mathtools,thm-restate}
\usepackage{hyperref}

\newcommand{\ignore}[1]{}

\everymath{\displaystyle}

\title{Incentive Compatible Two Player Cake Cutting}
\author{Avishay Maya\thanks{
Hebrew University of Jerusalem, School of Computer Science and
Engineering, avishay.maya@mail.huji.ac.il.  Research partially supported by a grant from the Israeli Science Foundation
and by a Google grant on Electronic Markets and Auctions.} \and Noam Nisan\thanks{
Microsoft Research and Hebrew University of Jerusalem, noam@cs.huji.ac.il. 
Research at the Hebrew University partially supported by a grant from the Israeli Science Foundation
and by a Google grant on Electronic Markets and Auctions.}}
\institute{}

\begin{document}
\maketitle
\begin{abstract}
We characterize methods of dividing a cake between two bidders in a way that is
incentive-compatible and Pareto-efficient.
In our cake cutting model, each
bidder desires a subset of the cake (with a uniform value over this subset),
and is allocated some subset.
Our characterization proceeds via reducing to a simple one-dimensional version
of the problem, and yields, for example, a tight bound on the social
welfare achievable.
\end{abstract}


\section{Introduction}
The question of allocating resources
among multiple people is one
of the most basic questions that humans have been studying.  At this level of
generality one may say that most of the economic theory is devoted to this
problem, as well as other fields of study.
One class of scenarios of this form, with an enormous amount of literature, goes
by the name of ``cake cutting''.
In this type of scenario the goods are modeled as the (infinitely divisible)
unit interval (the cake), the preferences as (measurable) valuation functions
on the cake and the allocation as a partition of the cake. Many variants of
this model have been considered and the usual goals are various notions of
fairness and efficiency. See, e.g. \cite{FairDivision} for an introduction.

Recently the research community has started looking at such models from a
mechanism design point of view, i.e., considering the incentives of the players.
From this perspective, players act {\em rationally} to maximize their utility
and will thus ``tell'' the cake cutting algorithm whatever will make it maximize
their own piece's value. In the simplest form\footnote{Which by the revelation
mechanisms is really without loss of generality since an arbitrary one, when analyzed at equilibrium, may be
converted to an incentive-compatible one where truth is an equilibrium.} we would ask for an``incentive
compatible'' (equivalently, truthful or strategy-proof) cake cutting allocation mechanism where each bidder
always maximizes his utility by reporting his true valuation.


Several recent papers have designed incentive-compatible cake cutting
mechanisms. For example, in \cite{TruthJustice} an incentive-compatible,
envy-free, Pareto-efficient, and proportional cake cutting mechanism is obtained
for the model where player valuations are ``uniform'': each player $i$
desires a subset $S_i$ of the items, and has a uniform value over this subset with the
total value of each player normalized to 1.\footnote{For purposes of efficient
computation, it is also required that the sets would be given as a finite
collection of intervals.}  In \cite{MosselTamuz}, an incentive-compatible, proportional, and
Pareto-efficient mechanism is constructed for the case of arbitrary (not
necessarily uniform) preferences.
A ``randomized'' cake cutting mechanism that is truthful in
expectation with better guarantees is also provided in that paper.

In this paper we seek to characterize incentive-compatible cake cutting
mechanisms, and show bounds on possible performance measures.
As our model has no ``money'' (i.e. no transferable utilities) the standard
tools of mechanism design with quasi-linear utilities (such as
Vickrey-Clarke-Groves \cite{Vic61,Cla71,Gro73} or Myerson \cite{Mye81}) do
not apply. In this sense our work lies within the framework of approximate mechanism design without money, advocated, e.g., by \cite{AFPT11,PT09}. 
As opposed to most of the cake cutting literature, we focus solely on incentive
compatibility and efficiency and do not consider notions of fairness.  As our results are
mostly ``negative'', this only strengthens them.  We should mention that the
positive results that we provide, i.e. the mechanisms that have the ``best''
properties among all incentive-compatible ones, turn out to also be envy-free.

Our general model, following that of \cite{TruthJustice}, considers an
infinitely-divisible atom-less cake and considers only the restricted class of
uniform player valuations.\footnote{Again, as our results are mostly ``negative'' this limited setting
strengthens them.} Formally, the ``cake'' is modeled as the real interval
$[0,1]$, each player desires a (measurable) set $A \subseteq [0,1]$ and his valuation is
uniform over that set (and normalized to 1):  $V_A(S) = |S \cap A|/|A|$, where
$| \cdot |$ specifies the usual Lebesgue measure.
We restrict ourselves to ``non-wasteful'' mechanisms, where no piece that is
desired by some player may be left unallocated and no piece is allocated to a
player that does not want it (this is essentially equivalent to Pareto-efficiency of the
outcome.\footnote{The inessential technical difference is
detailed in the next section.  While it does not {\em seem} that leaving pieces of the cake
unallocated can be useful, whether this is really the case remains open.}) We
restrict ourselves to the case of two players.
Thus a non-wasteful mechanism accepts as input the sets $A$ and $B$ desired by
the two players and returns two disjoint (measurable) sets $C = C(A,B) \subseteq
A$ and $D = D(A,B) \subseteq B$.  In this case, the first player's utility is
given by $V_A(C)=|C|/|A|$ and the second's by $V_B(D)=|D|/|B|$.  A mechanism is
called ``incentive-compatible'' if for every $A,B$ and $A'$ we get that
$V_A(C(A,B)) \ge V_A(C(A',B))$ and similarly for the second player.

As a tool for studying this model, we introduce a simple, one-dimensional ``aligned'' model.
In the aligned model we first restrict the possible player valuations:
the first player desires the sub-interval $A=[0,a]$ and the second player
desires the sub-interval $B=[1-b,1]$. This is interesting when $1-b < a$ in which case
the question is how to allocate the overlap $[1-b,a]$ between the players.  We
then also restrict the allowed allocation by the mechanism: the first player
must be allocated an interval $C=[0,c]$ and the second an interval $D=[1-d,1]$.
Thus, in the aligned model the input is fully specified by its lengths $a$,$b$,
the output by its lengths $c$,$d$, and a mechanism is a pair of real valued
functions $f=(c(a,b),d(a,b))$.
It turns out that these two restrictions
offset each other in some sense, allowing us to convert mechanisms between the
two models.  As the aligned model is really single-dimensional, we are able to fully characterize
incentive-compatible mechanisms in it, a characterization that then has strong
implications in the general model as well.


\begin{theorem} (Characterization of Aligned Model)
A non-wasteful deterministic mechanism for
two-players in the aligned model
is incentive-compatible if and only if
it is from the following family, characterized by $0\le \theta \le 1$:
the allocation gives the first player the interval
$\left[ 0,\min \left\{a,\max \left\{1-b,\theta\right\}\right\}\right]$
while the second player gets the interval
$\left[ 1 - \min \left\{b,\max \left\{ 1-a, 1-\theta\right\}\right\}, 1\right]$.
\end{theorem}

This characterization holds regardless of any issues of fairness, and the only
mechanism in this family that is fair in any sense is that with
$\theta=\frac{1}{2}$ which gives envy-freeness and turns out to be equivalent
to the mechanism of \cite{TruthJustice} for the case of two players.
This tight characterization in the aligned model allows the calculation of the
best achievable results -- under any desired performance measure --
for incentive-compatible mechanisms. Specifically,
we are interested in performance measures that depend on
relative lengths of demands and allocations, formally on the set
of 4-tuples $(\alpha,\beta,\gamma,\delta)$ where $\alpha=|A|/|A \cup
B|$, $\beta=|B|/|A \cup B|$, $\gamma=|C|/|A \cup B|$, and $\delta=|D|/|A \cup
B|$.\footnote{Note that as $|A \cap B|/|A \cup B|=\alpha+\beta-1$, $C
\subseteq A$, $D \subseteq B$, $C \cap D = \emptyset$, and $A\cup B=C\cup D$ we
have all the information regarding the sizes in the Venn diagram.}
A typical performance measure of this form is the
competitive ratio for social welfare: the worst case ratio between the social welfare achieved by the
mechanism (which is $\gamma/\alpha+\delta/\beta$) and that achieved at the optimal allocation (which turns out to be
$1+(1-\min\{\alpha,\beta\})/\max\{\alpha,\beta\}$).
Many other variants can be considered, such as looking at other aggregations of
the two players' utility (e.g. $\min\{\gamma/\alpha,\delta/\beta\}$ or
$\log(\gamma/\alpha)+\log(\delta/\beta)$), assigning different weights to the
different players, using a different comparison benchmark (e.g. the one
splitting the intersection equally), using additive regret rather than
multiplicative ratio, etc.

We prove the following reductions, which preserve the 4-tuples of
ratios $(\alpha,\beta,\gamma,\delta)$, between these models.

\begin{theorem} (Reduction Between Models)
\begin{enumerate}
	\item Let $f=(c(a,b),d(a,b))$ be an incentive-compatible and non-wasteful
	mechanism in the aligned model. There exists an incentive-compatible and
	non-wasteful mechanism $F=(C(A,B),D(A,B))$ in the general model such that for
	all $A,B$: $|C(A,B)|/|A \cup B| = c(a,b)$ and  $|D(A,B)|/|A \cup B|=d(a,b)$
	where $a=|A|/|A \cup B|$ and $b=|B|/|A \cup B|$.

	\item Let $F=(C(A,B),D(A,B))$ be an incentive-compatible and non-wasteful
	mechanism in the general model. There exists an incentive-compatible and
	non-wasteful mechanism $f=(c(a,b),d(a,b))$ in the aligned model such that for
	all $a,b$ there exist $A,B$ such that $|A|=a$, $|B|=b$, $c(a,b)=|C(A,B)|$ and
	$d(a,b)=|D(A,B)|$, and furthermore whenever $a+b \ge 1$ we have that $A \cup
	B=[0,1]$.
\end{enumerate}
\end{theorem}


These two reductions imply that while the general
model may be (and actually is) richer, this richness cannot buy anything in
terms of performance -- for any notion of performance that depends on relative
lengths of bids and allocations. For every mechanism with a certain performance
level in the general model there exists a mechanism with the same performance
level in the aligned model and vice-versa. 

Thus our characterization in the aligned model
implies the same bounds on performance 
in the general model as well.  For example, in the aligned model, 
one may easily calculate that
at most a fraction of $(8-4\sqrt{3})^{-1}\approx 0.93$ of social welfare can
be extracted by any mechanism in the characterized family, and this competitive ratio is
in fact obtained by the envy-free mechanism with $\theta=\frac{1}{2}$.  The reductions
imply that this same bound also applies to mechanisms in the general model.  This
ratio may thus be termed ``the price of truthfulness'' in this setting.  A
complementary result appears in \cite{CKKK09}, where the ``price
of fairness'' is studied, comparing envy-free allocations to general ones, and
obtaining the same numeric bound on the fraction of the optimal welfare that can
be extracted by any {\em envy-free} allocation.
Our results do not require any notion of fairness, but instead show
that incentive-compatibility by itself implies this bound.
In fact, for the special case of social welfare we also provide a direct proof for this 
bound, a proof that also applies to {\em randomized} mechanisms.

\begin{theorem} (Price of Truthfulness)
Any deterministic or randomized incentive-compatible mechanism for
cake cutting for two-players in the general model, achieves at most a
$(8-4\sqrt{3})^{-1}\approx 0.93$ fraction of the optimal welfare for some
player valuations.
\end{theorem}

It should be noted that this is tight, as indeed the deterministic mechanism of
\cite{TruthJustice} achieves this ratio when restricted to two players.

The paper is structured as follows: in section 2 we present our two models,
the general one and the aligned one. Section 3 provides the characterization of
the aligned model, and section 4 shows the reductions between the models.  In
section 5 we provide a direct proof of the price of truthfulness result for a
randomized mechanism.


\section{Models}
\subsection{The General Model}







Our model has two players each desiring a measurable subset of $[0,1]$.  We will denote by $A \subseteq [0,1]$ the
set desired player I and by $B\subseteq [0,1]$ the set desired by the player II.  We view $A$ and $B$ as private information.
Everything else is common knowledge.  The players will be assigned disjoint measurable 
subsets, $C\subseteq [0,1]$ to player I and $D\subseteq [0,1]$ to player II.
We assume that player valuations are uniform over the subsets they desire and normalized to 1.

\begin{definition}
The \emph{valuation} of a player who desires subset $A \subseteq [0,1]$ for a subset $C\subseteq[0,1]$ 
is $V_A(C)=|C\cap A|/|A|$, where $| \cdot |$ specifies
the Lebesgue measure.
\end{definition}

\begin{definition} 
A \emph{mechanism} is a function which divides the cake between the two
players. The function receives as inputs two measurable subsets of $[0,1]$:
$A$ and $B$ (the demands of the players), and outputs two disjoint measurable
subsets of $[0,1]$, $C$ and $D$, where $C$ is the subset that player I receives
and $D$ is the subset that player II receives.

We denote a mechanism by $F(A,B)=(C(A,B),D(A,B))$, where $C(\cdot),D(\cdot)$ 
denote the functions that determine the allocations to the two players, respectively, and must satisfy
$C(A,B)\cap D(A,B) = \emptyset$ for all $A,B$.
\end{definition}

Our point of view is that the two players are strategic, aiming to maximize their valuation and since $A$ and $B$ are private information 
the players may ``lie'' to the mechanism regarding their real interest in
the cake if that may give them an
allocation with a higher valuation for them.

\begin{definition}
$F=(C(A,B),D(A,B))$ is called {\em incentive-compatible} if none of the players
can gain by declaring a subset which is different from the real subset he is
interested in.
Formally, for all $A,B,A'$: $V_A(C(A,B))\geq V_A(C(A',B))$ and similarly
for the second player: for all $A,B,B'$: $V_B(D(A,B))\geq V_B(D(A,B'))$.
\end{definition}

\begin{definition}
A mechanism $F=(C(A,B),D(A,B))$ is said to be {\em Pareto-efficient} if for
every input $A,B$ and the corresponding allocation made by the mechanism
$C(A,B),D(A,B)$, any other possible allocation $C',D'$ can not be strictly better for one
of the players and at least as good for the other.
\end{definition}

Note that two possible allocations $C,D$ and $C',D'$, which differ only in
the division of areas which none of the players is interested in, are equivalent
in the eyes of the players. Therefore, we would use a specific
Pareto-efficient allocation -- a non-wasteful allocation, in which
pieces of the cake that neither of the players demanded will not be allocated.

\begin{definition}
A mechanism $F=(C(A,B),D(A,B))$ is called {\em non-wasteful} if for every $A,B$ we have that 
$C(A,B)\subseteq A$, $D(A,B)\subseteq B$, and $C(A,B)\cup D(A,B)=A\cup B$.
\end{definition}

\begin{proposition}
Every non-wasteful mechanism is Pareto-efficient.  Every Pareto-efficient mechanism $F=(C(A,B),(D(A,B))$
can be converted to an equivalent non-wasteful one by defining $C'(A,B)=C(A,B)\cap A$ and $D'(A,B)=D(A,B) \cap B$.
\end{proposition}

Thus any analysis of non-wasteful mechanisms directly implies a similar one for Pareto-efficient ones, as do all our results in this paper.
For a non-wasteful mechanism the valuations of the players are simply
$|C|/|A|$ for player I and $|D|/|B|$ for player II.




Although we do not deal directly with the envy-freeness of mechanisms, a
mechanism that is described in this paper has this property, as described
below.

\begin{definition}
$F=(C(A,B),D(A,B))$ is called {\em envy-free} if each player
weakly prefers the piece he received to the piece the other player received.
Formally, for all $A,B$: $V_A(C(A,B))\geq V_A(D(A,B))$ and similarly
for the second player, for all $A,B$: $V_B(D(A,B))\geq V_B(C(A,B))$.
\end{definition}

\subsection{The Aligned Model}
A special case of the above general model is called the {\em
aligned} model.  The model makes two specializing assumptions, one on player valuations, and the other on mechanism allocations:
\begin{enumerate}
\item
The two players are interested in
subsets of the form $[0,a]$ for player I and $[1-b,1]$ for player II.  
\item
The mechanism must divide te cake so that
player I and player II would receive subsets of the form $[0,c]$ and
$[1-d,1]$ respectively.
\end{enumerate}

In the aligned model we denote a mechanism as $f(a,b)=(c(a,b),d(a,b))$,
Where $c,d$ are in fact functions $c,d\colon \mathbb{R}^+ \times
\mathbb{R}^+ \to \mathbb{R}^+$, such that for all $a,b$:
$~c(a,b)+d(a,b)\leq 1$.

\subsection{The Price of Truthfulness}
As noted in the introduction, using the two reductions that will be proved
in section 4, it is possible to study a family of performance measures for the
aligned model and conclude from that implications for the general models. For
example, one of these performance measures is the Price of Truthfulness.

\begin{definition}
The {\em social welfare} of a mechanism $F=(C(A,B),D(A,B))$
on input $A,B$, denoted by $SW_{F}(A,B)$, is 
$SW_{F}(A,B)=V_A(C(A,B))+V_B(D(A,B))$.
\end{definition}

\begin{definition}
Denoted by $SW_{max}(A,B)$ is the sum of valuations of the two players in the
allocation that maximizes social welfare: 
$SW_{max}(A,B) = \max _{F} SW_{F}(A,B)$.
\end{definition}

\begin{definition}
The \emph{competitive ratio} for social welfare of a mechanism $F$ 
is $\eta _{F} = \min _{A,B} \eta _{F}(A,B)$, where 
$\eta _{F}(A,B) = \frac{SW_{F}(A,B)}{SW_{max}(A,B)}$.
\end{definition}

Similar to the price of anarchy, the \emph{price of truthfulness} is
the highest possible competitive ratio of a truthful mechanism.
Formally:

\begin{definition}
The price of truthfulness is $PoT\equiv\max_F \eta_F$, where $F$ ranges over all non-wasteful truthful mechanisms.
\end{definition}



\subsection{Randomized Mechanisms}
In the last part of our paper we will also consider randomized mechanisms.  
For the purposes of this paper, one may either consider those as a probability distribution over deterministic mechanisms,
or allow the mechanism's allocation $(C,D)$ to be a random variable.
\begin{definition}
For a randomized mechanism
$F$, the above definitions are extended by replacing
$SW_F(A,B)$ by $\mathbb{E} \left[ SW_F(A,B)\right]$ where the expectation is
over the random choices made by the mechanism.
\end{definition}

\section{The Aligned Model}
\subsection{Characterization of the Aligned Model}
\begin{theorem}\label{aligned_theorem}(Characterization of Aligned Model)
A non-wasteful deterministic mechanism for
two-players in the aligned model
is incentive-compatible if and only if
it is from the following family, characterized by $0\le \theta \le 1$:
the allocation gives the first player the interval
$\left[ 0,\min \left\{a,\max \left\{1-b,\theta\right\}\right\}\right]$
while the second player gets the interval
$\left[ 1 - \min \left\{b,\max \left\{ 1-a, 1-\theta\right\}\right\}, 1\right]$.
\end{theorem}

The remainder of this subsection is a proof of the above theorem.

\medskip 

Assume $f(a,b)=(c(a,b),d(a,b))$ is a non-wasteful
incentive-compatible deterministic mechanism for two-players in the aligned
model.

In case $a+b\leq 1$, there is no overlap between the demands of the players
which are aligned to the sides. Therefore, from non-wastefulness,
the mechanism would have to give each player all of his demand (and that is
clearly incentive-compatible and deterministic). We can also notice that this
scenario matches the expressions for the pieces of the cake allocated to the
players from the theorem, regardless of $\theta$.

During the rest of this proof, we will assume that there is an
overlap between the demands of the player, i.e. $a+b > 1$.

\begin{definition}
For the mechanism $f(a,b)=(c(a,b),d(a,b))$ and a fixed demand $b$ for player
II, we will denote by $c_b(a)$ the function $c(a,b)$, which
determines the size of the piece that the mechanism gives to player I according
to his demands $a$.
In a similar way $d_a(b)$ is also defined.
\end{definition}

\begin{lemma}
For every $b$, the function $c_b(a)$ of the mechanism $f(a,b)$ is non
decreasing and Lipschitz continuous (with a Lipschitz constant $K=1$).
\end{lemma}

\begin{proof}
For $a<a'$, say that $c_b(a)>c_b(a')$, then from non-wastefulness, $a'>a\geq
c_b(a)>c_b(a')$. Therefore, if player I's real interest is a piece of size $a'$,
he can gain strictly more by demanding $a$ instead. He would receive not only a
larger piece of the cake, but also a larger piece of his interest, due to the
alignment of the piece to the side.
That stands in contradiction to the incentive-compatibility of
the mechanism. Hence, $c_b(a)\leq c_b(a')$, meaning that $c_b(a)$ is non
decreasing.

Furthermore, for $a<a'$, $c_b(a')-c_b(a)\leq a'-a$.
Otherwise, if $c_b(a')-c_b(a)>a'-a$,
this means that $c_b(a')-a'+a>c_b(a)$. Since the mechanism is non-wasteful,
$c_b(a')\leq a'$, and therefore $a>c_b(a)$. In such a case, if player I's real
interest is of size $a$, he will not receive all of his demand. Therefore, he
might lie and demand $a'$ instead. By asking for $a'$ he would receive a larger
piece ($c_b(a')-a'+a>c_b(a)\Rightarrow c_b(a')>c_b(a)$), which because of
the alignment, has a larger intersection with his real interest. Again, this
contradicts the incentive-compatibility of the mechanism.

We have that $c_b(a)$ is Lipschitz continuous (with a Lipschitz constant $K=1$).
\end{proof}

Therefore $c_b(a)$ is continuous. Hence, in the interval $[0,1]$ it must
attain a maximum value, and the following quantities are well defined.

\begin{definition}
$\mu(b)$ is the maximal piece size that player I can receive, when
player II demands a piece of size $b$.
Formally, $\mu(b)\equiv \max_{a}{c_b(a)}$.

In the same way $\nu(a)\equiv \max_{b}{d_a(b)}$ is defined for player II.
\end{definition}

\begin{definition}
We will denote by $a_{m}$ the minimal $a$ for which
$c_b(a_{m})=\mu(b)$.
\end{definition}

\begin{lemma}
For the mechanism $f(a,b)$ as mentioned, for every $b$:
\[
c_b(a)=	\begin{dcases*}
			a 	 	& for $a < \mu(b)$\\
			\mu(b) 	& for $a \geq\mu(b)$
		\end{dcases*}
	= \min\left\{a,\mu(b)\right\}
\]

(see Figure \ref{cGraph})
\end{lemma}

\begin{proof}
{\text

}
\begin{itemize}
  \item For $a<a_{m}$, $c_b(a)$ can not be larger than $a$, because of the
  non-wastefulness of $f(a,b)$. If $c_b(a)<a$, then player I,
  whose real interest is of size $a$, does not receive all of his interest and
  therefore would prefer to lie and ask for $a_m$. Since $a<a_{m}$, by
  definition of $a_m$, $c_b(a)<c_b(a_m)$. Not only would player I receive
  a strictly larger piece by lying, since the piece is aligned to the side, he
  would also receive a strictly larger piece of his real interest. This stands
  in contradiction to the incentive-compatibility of $f$. Therefore, $c_b(a)=a$.
  \item For $a>a_m$, since $c_b(a)$ is non-decreasing, $c_b(a)\geq c_b(a_m)$. It
  is also known that $c_b(a_m)=\mu(b)$ is the maximal value of $c_b(a)$.
  Therefore, $c_b(a)=\mu(b)$.
  \item We showed that for $a<a_{m}$, $c_b(a)=a$, hence $c_b(a_m)=a_m$ by
  continuity. Since $c_b(a_m)=\mu(b)$, $a_m=\mu(b)$.

%
%
%
\end{itemize}
Putting everything together, we get that $c_b(a)= \min\left\{a,\mu(b)\right\}$.

\end{proof}

\begin{remark}
Characterization of player II's piece size for a fixed $a$ can be done in the
same way to obtain $d_a(b)= \min\left\{b,\nu(a)\right\}$.
\end{remark}


Now, we can continue to the characterization of the function $\mu(b)$.

\begin{remark}
We should notice that $\mu(1)+\nu(1)=1$ (from
non-wastefulness, in case both players want the whole cake we should
divide the whole cake).
\end{remark}

\begin{lemma}
The function $\mu(b)$ must be of the form:
\[
 \mu(b) = 	\begin{dcases*}
				1- b  	& for $b < 1-\theta$\\
				\theta 	& for $b \geq 1-\theta$
			\end{dcases*}
    	=\max\left\{1-b,\theta\right\}
\]
(For $\theta \in [0,1]$).
(see Figure \ref{muGraph})
\end{lemma}

\begin{proof}
According to the function $c_b(a)$, which we found earlier, the size of
the piece that player I receives is $\min \{a,\mu(b)\}$. As mentioned in
the beginning of the subsection, we assume that $a+b > 1$. As we
are examining the aligned model, the mechanism should divide the
whole interval $[0,1]$.
Therefore, player II would receive $1-\min \{a,\mu(b)\}$.
We also know that the form of the function $d_a(b)$
resembles the form of $c_b(a)$ and that means that the size of the piece
that player II receives is $\min \{b,\nu(a)\}$. Combined together:

\[
1-\min \{a,\mu(b)\}=\min \{b,\nu(a)\}=\
            \begin{dcases*}
            	b 		& for $b < \nu(a)$\\
            	\nu(a) 	& for $b \geq \nu(a)$
            \end{dcases*}\
\]

\[
\xRightarrow\
\min \{a,\mu(b)\}=\
				\begin{dcases*}
            		1-b  	& for $b < \nu(a)$\\
            		1-\nu(a)& for $b \geq \nu(a)$
            	\end{dcases*}\
\]

Let us look at the last equation for $a=1$:
\[
\mu(b)=\min \{1,\mu(b)\}=\
            	\begin{dcases*}
            		1-b  	& for $b < \nu(1)$\\
            		1-\nu(1)& for $b \geq \nu(1)$
            	\end{dcases*}\
\]
We also know that $\mu(b)$ does not depend on $a$. Therefore, the last
statement is true in general and not only for $a=1$.
We showed previously that $\mu(1)+\nu(1)=1$. Let us denote
$\theta\equiv\mu(1)=1-\nu(1)$, and rewrite $\mu(b)$ ($\nu(a)$ can be found in a
similar way):
\[
\mu(b)=\
         	\begin{dcases*}
            	1-b  	& for $b < 1 - \theta$\\
            	\theta 	& for $b \geq 1 - \theta$
            \end{dcases*}\
    	=\max\{1-b,\theta\}
\]
\[
\nu(a)=\
            \begin{dcases*}
            	1-a  	& for $a < \theta$\\
            	1-\theta& for $a \geq \theta$
            \end{dcases*}\
    	=\max\{1-a,1-\theta\}
\]

\end{proof}

\begin{figure}[h]
\begin{minipage}[b]{0.5\linewidth}
	\centering
	\includegraphics[width=2.0in]{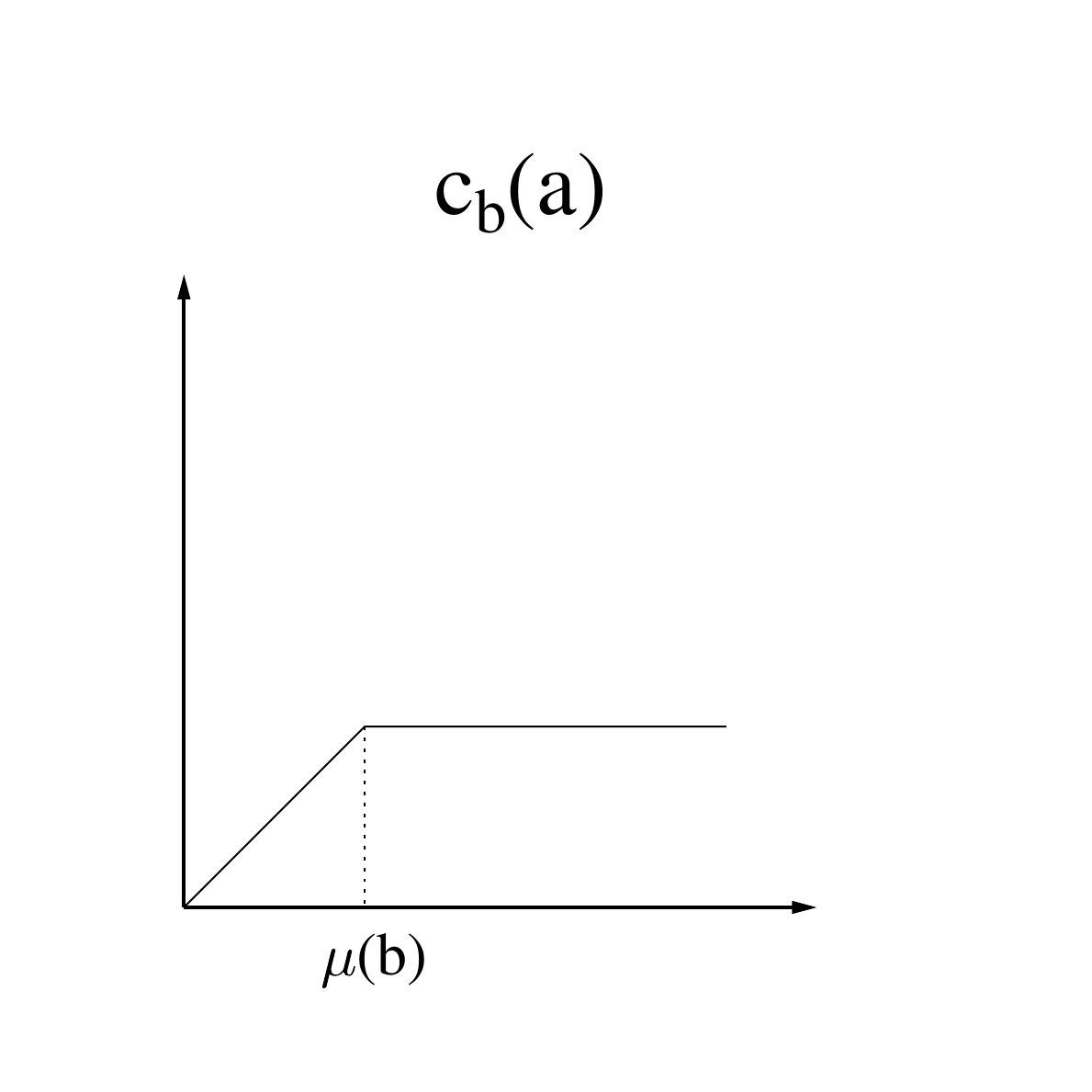}
    \caption{The size of the allocation for player I as a function of his demand
    (for a constant $b$)}
    \label{cGraph}
\end{minipage}
\hspace{0.5cm}
\begin{minipage}[b]{0.5\linewidth}
	\centering
	\includegraphics[width=2.0in]{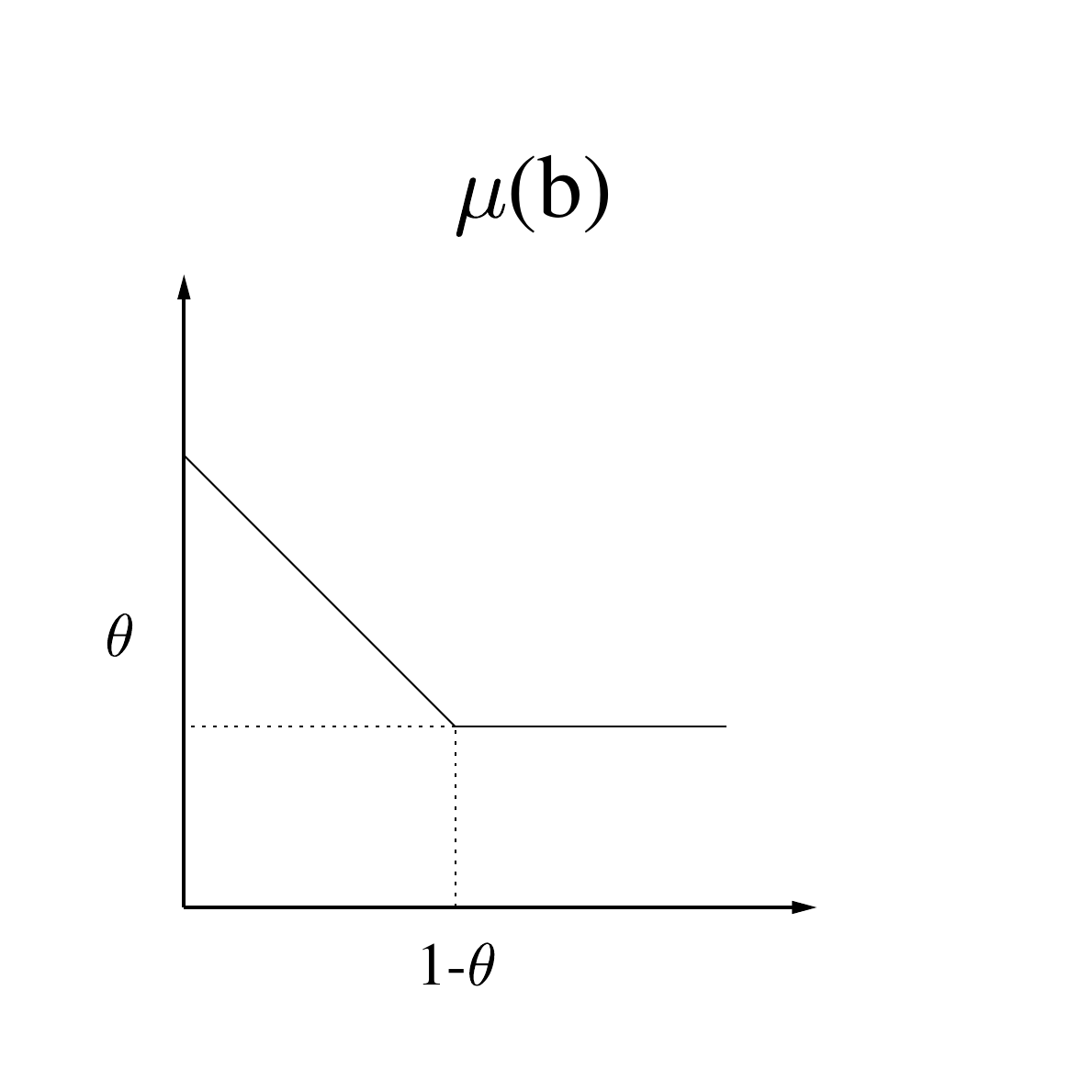}
    \caption{The value in which the graph $c_b(a)$ (for a specific $b$) turns 
    constant, as a function of that $b$)}
    \label{muGraph}
\end{minipage}
\end{figure}


If we insert those $\mu(b)$ and $\nu(a)$ into the expressions for
$c_b(a)$ and $d_a(b)$ that we found earlier, we would get that
$c(a,b)=\min\{a,\max\{1-b,\theta\}\}$ and
$d(a,b)=\min\{b,\max\{1-a,1-\theta\}\}$, as in the statement of the theorem.

In the opposite direction, it can be noticed that the allocation is
deterministic. Furthermore, for all values of $a,b$ and $\theta$, each of the
players either receives all of his demand, or a maximal value which
depends only on the other player. Therefore, he cannot gain by lying.
Moreover, $c(a,b)+d(a,b)=\min\{a+b,1\}$, and because of the alignment of the
interests and allocations, this type of allocation is non-wasteful.

We conclude that this is in fact the family of all possible mechanisms. We will
denote by $f_{\theta}$ the mechanism with the parameter $\theta$ from that family.

\subsection{Social Welfare in the Aligned Model}
\begin{restatable}{theorem}{thetaHalf}\label{theta_half}
The non-wasteful and incentive-compatible deterministic mechanism
$f_{\frac{1}{2}}$ for the aligned model achieves
$\eta_{f_{\frac{1}{2}}}=(8-4\sqrt{3})^{-1}\approx 0.93$.
\end{restatable}

This theorem is proved in the appendix. 

\begin{remark}\label{half_best_theta}
It is possible to prove that for every other mechanism $f_{\theta}$ from the
family of mechanisms for the aligned model,
$\eta_{f_{\theta}}<\eta_{f_{\frac{1}{2}}}$, by checking the value of
$\eta_{f_{\theta}}$ for two possible sets of inputs:
$\tilde{a}=1,\tilde{b}=\sqrt{3}-1$ and
$\tilde{a}=\sqrt{3}-1,\tilde{b}=1$. We will prove a stronger theorem in the last
section.

Moreover, it can be noticed that $\theta=\frac{1}{2}$ is the only $\theta$
for which $f_{\theta}$ is envy-free.
\end{remark}


\section{Reductions}
\subsection{Reduction From the Aligned to the General Model}
\begin{theorem}
Let $f=(c(a,b),d(a,b))$ be an incentive-compatible and non-wasteful
mechanism for the aligned model. There exists an incentive-compatible and
non-wasteful mechanism $F=(C(A,B),D(A,B))$ for the general model such that for
all $A,B$: $|C(A,B)|/|A \cup B| = c(a,b)$ and  $|D(A,B)|/|A \cup B|=d(a,b)$
where $a\equiv|A|/|A \cup B|$ and $b\equiv|B|/|A \cup B|$.
\end{theorem}

Note that from the properties of $f$ it has to be from the family of
mechanisms described in the previous section. Therefore there is a $\theta$
such that $f$ is $f_{\theta}$.


For that $f_{\theta}$, we will define mechanism
$F(A,B)$ as follows:
\begin{itemize}
  \item  Use the mechanism $f_{\theta}$ to calculate the size of the players'
  allocations $(c(a,b),d(a,b))$ when:\footnote{The division by
  $|A \cup B|$ in this phase is a normalization of the original demands over a
  full $[0,1]$ interval.}
    \begin{itemize}
        \item $a=\frac{|A|}{|A \cup B|}$, meaning player I demands the section
        $[0,\frac{|A|}{|A \cup B|}]$.
        \item $b=\frac{|B|}{|A \cup B|}$, meaning player II demands the section
        $[1-\frac{|B|}{|A \cup B|},1]$.
    \end{itemize}
  \item  Calculate
  $|C(A,B)| \equiv c(a,b) \cdot|A \cup B|$ and
  $|D(A,B)| \equiv d(a,b) \cdot|A \cup B|$. \footnote{A normalization of the
  results back to the original interval.}
  \item  Give player I pieces in a total size of $|C(A,B)|$ and Player II pieces
  in a total size of $|D(A,B)|$. For each of them -- start at first from giving
  the cake intervals that only he asked for, then move to intervals in the joint
  area.
\end{itemize}

The size of the piece that mechanism $F$ would allocate to player I is:
$|C(A,B)|=|A \cup B|\cdot\min\{\frac{|A|}{|A \cup B|},\max\{1-\frac{|B|}{|A
\cup B|},\theta\}\}=\min\{|A|,\max\{|A \cup B|-|B|,\theta\cdot|A \cup
B|\}\}=\min\{|A|,\max\{|A\setminus B|,\theta\cdot|A \cup B|\}\}$.
In a similar way we can get the expression for the size of player II's
piece.

\begin{lemma}
$F$ is non-wasteful.
\end{lemma}

\begin{proof}
The mechanism assigns two pieces with total size of 
$|C(A,B)|+|D(A,B)|=(c(a,b) + d(a,b))\cdot|A \cup
B|\underset{a+b\ge1\rightarrow c+d=1}{=}|A \cup B|$, meaning the total size
that was assigned is equal to the total requested size.
Moreover, $c(a,b)\le a=\frac{|A|}{|A\cup B|}$, therefore
$|C(A,B)|\le|A|$ and
in the same way $|D(A,B)|\le|B|$.
This means that the mechanism gives each player no more than his demand.
Therefore, it is possible to construct the player's allocation only from
intervals he has asked for. Since the allocation of those pieces starts with
intervals that only one player asked for and because
the total size allocated is $|A \cup B|$, the division is non-wasteful.
\end{proof}

\begin{restatable}{lemma}{reductionA}
$F$ is incentive-compatible.
\end{restatable}


In this proof we examine a general subset $A_1$ which differs from the real
interest of player I, $A$. We look at the symmetric difference between those two subsets,
divide it into 4 disjoint sets, and one after the other show that zeroing the
size of a set cannot damage the player. Therefore, he has no interest to lie.
This lemma is fully proved in the appendix.
\medskip 

Concluding, the mechanism $F$ meets the demands of the theorem, thus completing
the proof.

\medskip 

Say we choose $f$ and examine the matching mechanism $F$, as described.
If the inputs for mechanism $F$ are $A,B$, we can look at the 4-tuple of ratios
created by $F$: $\left(\frac{|A|}{|A\cup B|},\frac{|B|}{|A\cup
B|},\frac{|C|}{|A\cup B|},\frac{|D|}{|A\cup B|}\right)$.
The above reduction shows that the inputs $a=\frac{|A|}{|A\cup
B|},b=\frac{|B|}{|A\cup B|}$ for mechanism $f$ will result in the output
$c=\frac{|C|}{|A\cup B|},d=\frac{|D|}{|A\cup B|}$. Since $a+b=\frac{|A|}{|A\cup
B|}+\frac{|B|}{|A\cup B|}\geq 1$ and since the requests are aligned to different
sides, the total demand made by the two players is of size 1. Therefore, in this
case, the 4-tuple of ratios is $(a,b,c,d)$, which is identical to the 4-tuple
that was obtained by $F$ on the inputs $A,B$.


%
\subsection{Reduction From the General to the Aligned
Model}\label{second_reduction}
\begin{theorem}
Let $F=(C(A,B),D(A,B))$ be an incentive-compatible and non-wasteful
mechanism for the general model. There exists an incentive-compatible and
non-wasteful mechanism $f=(c(a,b),d(a,b))$ for the aligned model, such that for
all $a,b$ there exist $A,B$ such that $|A|=a$, $|B|=b$, $c(a,b)=|C(A,B)|$ and
$d(a,b)=|D(A,B)|$, and furthermore whenever $a+b \ge 1$ we have that $A \cup
B=[0,1]$.
\end{theorem}


For mechanism $F(A,B)$ as mentioned, we will define mechanism $f(a,b)$ as
follows:
\begin{itemize}
  \item Find the division made by $F$ in case both players want the whole cake:\\
  $F([0,1],[0,1])=(\tilde{C},\tilde{D})$. Since $F$ is non-wasteful, $\tilde{C}\uplus\tilde{D}=[0,1]$.
  \item Denote $c(a,b)\equiv \min \left\{a,\max \left\{1-b,|\tilde{C}|\right\}\right\}$
  \item Denote $d(a,b)\equiv \min \left\{b,\max \left\{1-a,1-|\tilde{C}|\right\}\right\}$
  \item Give players I and II pieces $\left[ 0,c(a,b)\right]$ and $\left[1-d(a,b),1 \right]$ respectively.
\end{itemize}

\begin{lemma}
$f=(c(a,b),d(a,b))$ is non-wasteful and incentive-compatible.
\end{lemma}

\begin{proof}
$f$ receives an aligned players' input and divides the cake into aligned pieces.
The size $|\tilde{C}|$ is between $0$ and $1$ (similar to $\theta$). The sizes
of the pieces that the players receive ($c(a,b)$ and $d(a,b)$) match the
family of mechanisms that was mentioned in the section about the aligned model,
for $\theta=|\tilde{C}|$.
Therefore, $f$ is, in fact, the mechanism $f_{|\tilde{C}|}$ from that family.
We already know that for aligned players' valuation function (as in
this case), mechanisms from that family are non-wasteful and
incentive-compatible.
\end{proof}

\begin{restatable}{lemma}{reductionB}
For $F=(C(A,B),D(A,B))$ and $f=(c(a,b),d(a,b))$ as defined,
for all $a,b$ there exists $A,B$ such that $|A|=a$, $|B|=b$, $c(a,b)=|C(A,B)|$
and $d(a,b)=|D(A,B)|$, and furthermore whenever $a+b \ge 1$ we have that $A \cup
B=[0,1]$.
\end{restatable}

This lemma is proved in the appendix.
\medskip 

Concluding, the mechanism $f$ meets the demands of the theorem, completing the
proof.
\medskip 

Say we choose $F$ and examine the matching mechanism $f$, as described. Denote
the inputs of mechanism $f$ as $a,b$. If $a+b\leq 1$, choosing
$A=[0,a],B=[1-b,1]$ as inputs for $F$ will result in each of the players 
receiving all of his demand, causing an identical 4-tuple of ratios for the
two mechanisms:
$\left(\frac{a}{a+b},\frac{b}{a+b},\frac{a}{a+b},\frac{b}{a+b}\right)$.
If $a+b>1$, the union of the players' demands is of size 1. The theorem shows
that there are $A,B$ such that $|A|=a,|B|=b,|C(A,B)|=c(a,b),|D(A,B)|=d(a,b)$
and furthermore, $|A\cup B|=1$.
Therefore, the ratio 4-tuples obtained by $f(a,b)$ and $F(A,B)$ (for the
specific $A$ and $B$ suggested in the theorem) are identical: $(a,b,c,d)$.



\section{The Price of Truthfulness}
As was mentioned in Remark \ref{half_best_theta},
it is possible to show that for any $0\leq\theta\leq1 ~,~
\theta\neq\frac{1}{2}$, the competitive ratio of the social welfare of the
mechanism $f_{\theta}$ (marked as $\eta_{f_{\theta}}$), is $<(8-4\sqrt{3})^{-1}$.
Using the two reductions from the last section, we can conclude
that there isn't a non-wasteful, incentive-compatible,
deterministic mechanism for the general model with higher $\eta$.
Moreover, Since $\eta_{f_{\frac{1}{2}}}=(8-4\sqrt{3})^{-1}$, there
is an incentive-compatible, non-wasteful deterministic mechanism $F$ for general
model\footnote{This mechanism is
$F$ that is generated by reduction \ref{second_reduction} using the
mechanism $f_{\frac{1}{2}}$.} which achieves $\eta _F =
(8-4\sqrt{3})^{-1}\approx 0.93$.

We will now prove a stronger claim - this upper bound still holds even if
the mechanism can be wasteful or randomized, as long as the valuation
functions are of the same form which we defined in the general model (actually,
the exact proof is even stronger and also works even if the players are limited
only to the aligned model's valuation functions).

\begin{theorem}(Price of Truthfulness)
Any deterministic or randomized incentive-compatible mechanism for
cake cutting for two-players in the general model, achieves at most a
$(8-4\sqrt{3})^{-1}\approx 0.93$ fraction of the optimal welfare for some
player valuations.
\end{theorem}

\begin {proof}
Say each of the two players' real demand is the whole cake: $[0,1]$. We will
denote by $p$ and $q$ the expected sizes of the pieces of cake that the
mechanism gives player I and player II in that case, respectively.
W.l.o.g we assume that player I received the (weakly) smaller piece, $p\leq q$
and since $p+q\leq1$, $p\leq\frac{1}{2}$.

Now, we will examine what happens if player I's demand is $A=[0,1-\tau]$ for
some $0\leq\tau\leq 1$, and player II's demand remains unchanged.
Intuitively, in order to maximize the social welfare, as a player
demands a smaller piece, the mechanism needs to give him a larger allocation (in
case he really asks for his real demand).
However, from incentive-compatibility, the size of the piece that player I will
receive can not be greater than $p$ (if it did, it would have been better for
him to lie in the previous case and ask for the smaller piece instead of the
whole cake). We denote by $p',q'$ the expected size of the pieces that the
players receive in that case.

Since $\eta _{F} = \min _{A,B} \frac{\mathbb{E}\left[SW_F(A,B)\right]}{SW_{max}(A,B)}$, and we are
checking only a specific subset of inputs (of
the form $A=[0,1-\tau],B=[0,1]$), we can say that for each of those $A,B$:


  $$
  \eta_F\leq
  \frac{\mathbb{E}\left[SW_F(A,B)\right]}{SW_{max}(A,B)}=
  \frac{\frac{p'}{1-\tau}+\frac{q'}{1}}{\frac{1-\tau}{1-\tau}+\frac{\tau}{1}}\underset{\underset{q'\leq 1-p'}{p'\leq p,1-\tau\leq 1}}{\leq}
  $$
  $$
  \frac{\frac{p}{1-\tau}+1-p}{1+\tau}=
  \frac{1+p\left(\frac{1}{1-\tau}-1\right)}{1+\tau}\underset{\underset{p\leq\frac{1}{2}}{\frac{1}{1-\tau}-1>0}}{\leq}
  \frac{\frac{\frac{1}{2}}{1-\tau}+\frac{1}{2}}{1+\tau}
  $$

The minimal value of this expression is $(8-4\sqrt{3})^{-1}$ at $\tau=2-\sqrt{3}$.

Therefore, $\eta_F\leq (8-4\sqrt{3})^{-1}$
\end{proof}

We remark again -- there exists a mechanism $F$, in the general model, which
achieves the bound for a mechanism in that model, $(8-4\sqrt{3})^{-1}\approx
0.93$. This is the price of truthfulness.

On top of being incentive compatible, this mechanism is also deterministic,
non-wasteful and envy-free. 


\bibliography{ICCakeCut}
\bibliographystyle{plain}

\appendix
\section{Social Welfare in the Aligned Model}

\thetaHalf*
\begin{proof}
In the cases in which $a+b\leq1$, as stated earlier, each of the
players gets all of his demand. Therefor, $SW_{f_{\frac{1}{2}}}=SW_{max}=2$, and
$\eta_{f_{\frac{1}{2}}}$ is at its maximal possible value ($1$).

In case that $a+b>1$:

$SW_{f_{\frac{1}{2}}}(a,b)=\frac{c(a,b)}{a}+\frac{1-c(a,b)}{b}=\frac{1}{a}[\min\{a,\max\{1-b,\frac{1}{2}\}\}]+\frac{1}{b}[1-\min\{a,\max\{1-b,\frac{1}{2}\}\}]=$

\[
  	\begin{dcases*}
    	1 + \frac{1-a}{b}  			& for $a < \max\{1-b,\frac{1}{2}\}$\\
    	\frac{1-b}{a} + 1  			& for $\frac{1}{2} < 1-b < a$\\
    	\frac{1}{2a}+\frac{1}{2b}	& for $1-b < \frac{1}{2} < a$
    \end{dcases*}
\]

\[
SW_{max}(a,b)=
    \begin{dcases*}
    	1 + \frac{1-a}{b}  & for $a < b$\\
    	\frac{1-b}{a} + 1  & for $b < a$
    \end{dcases*}
\]

Now, we will check the value of
$\eta_{f_{\frac{1}{2}}}(a,b)=\frac{SW_{f_{\frac{1}{2}}}(a,b)}{SW_{max}(a,b)}$
for each of the possible orders (from lowest to highest) of
$a,1-a,b,1-b,\frac{1}{2}$.
It is not possible that $a$ and $1-a$ would both be smaller or larger than
$\frac{1}{2}$ (because $0\leq a \leq1$). The same is true also for $b$
and $1-b$.
Moreover, $a$ and $b$, can not be both smaller than $\frac{1}{2}$,
because their sum should be $>1$.
Finally, from the same reason, it is also not possible that $b\leq 1-a$
or $a \leq 1-b$.
Therefore, there are only 4 possible permutations:

\begin{center}
    \begin{tabular}{ |c|c|c|c|c|}
    \hline
    \multicolumn{1}{|c|}{\#} &
    \multicolumn{1}{|c|}{order (low to high)} &
    \multicolumn{1}{|c|}{$\min\{a,\max\{1-b,\frac{1}{2}\}\}$} &
    \multicolumn{1}{|c|}{$SW_{f_{\frac{1}{2}}}$} &
    \multicolumn{1}{|c|}{$SW_{max}$} \\ [0.2cm]\hline
    &&&&\\[-0.3cm]
    1&$1-a,b,\frac{1}{2},1-b,a$ & $1-b$    		& $\frac{1-b}{a} + 1$          & $1+\frac{1-b}{a}$\\ [0.2cm]\hline
    &&&&\\[-0.3cm]
    2&$1-a,1-b,\frac{1}{2},b,a$ & $\frac{1}{2}$	& $\frac{1}{2a}+\frac{1}{2b}$  & $1+\frac{1-b}{a}$\\ [0.2cm]\hline
    &&&&\\[-0.3cm]
    3&$1-b,a,\frac{1}{2},1-a,b$ & $a$     		& $1 + \frac{1-a}{b}$          & $1+\frac{1-a}{b}$\\ [0.2cm]\hline
    &&&&\\[-0.3cm]
    4&$1-b,1-a,\frac{1}{2},a,b$ & $\frac{1}{2}$	& $\frac{1}{2a}+\frac{1}{2b}$  & $1+\frac{1-a}{b}$\\ [0.2cm]\hline
    \end{tabular}
\end{center}

In cases 1 and 3, $\eta_{f_{\frac{1}{2}}}(a,b)$ gets to its maximum (1).

In case 2: $\eta_{f_{\frac{1}{2}}}(a,b)=\frac{a+b}{2b(1-b+a)}$
and the minimum is at $\tilde{a}=1,\tilde{b}=\sqrt{3}-1$

In case 4: $\eta_{f_{\frac{1}{2}}}(a,b)=\frac{a+b}{2a(1-a+b)}$
and the minimum is at $\tilde{a}=\sqrt{3}-1,\tilde{b}=1$

In both cases 2 and 4, The value of
$\eta_{f_{\frac{1}{2}}}(a,b)$ at its minima is
$\frac{1}{8-4\sqrt{3}}\approx 0.93$, and this is $\eta_{f_{\frac{1}{2}}}$
\end{proof}

\section{Reduction From the Aligned to the General Model}
\reductionA*
\begin{proof}
Without loss of generality, we will show that player I can't benefit from
lying. Say that player I demanded some $A_1$, which differs from his real will
$A$.
It is possible to divide the symmetric difference between $A$ and $A_1$ into
four disjoint subsets:
\begin{itemize}
  \item $\Delta_1\equiv (A_1\setminus A) \cap B$ 
  \item $\Delta_2\equiv (A\setminus A_1) \cap B$ 
  \item $\Delta_3\equiv (A_1\setminus A) \cap \bar{B}$ 
  \item $\Delta_4\equiv (A\setminus A_1) \cap \bar{B}$ 
\end{itemize}

Note that
$\Delta_1$ and $\Delta_3$ are the subsets that are not in $A$ but were
added to $A_1$.
$\Delta_2$ and $\Delta_4$ are the subsets of $A$ that are not in $A_1$.

In this proof, we will show (one after the other) that zeroing
the size of those $\Delta$'s, can only raise the profit of player I (for every $A$ and
$A_1$). Therefore player I has no incentive to lie.

The total size of the pieces that player I would get from demanding $A_1$ is:
$C(A_1,B)=\min\{|A_1|,\max\{|A_1\setminus B|,\theta\cdot|A_1 \cup B|\}\}$. We
should note that although the piece that player I would receive from
mechanism $F$ must be included in $A_1$, it is not necessarily included in $A$.

We will define $A_2$ as $A_1$ without the subset that was added $\Delta_1$,
$A_2\equiv A_1\setminus\Delta_1$ and show that player I can only benefit from
demanding $A_2$ instead of $A_1$.
The size of the piece that player II would receive after demanding $A_2$
is: $C(A_2,B)=\min\{|A_2|,\max\{|A_2\setminus B|,\theta\cdot|A_2 \cup B|\}\}$.
If the minimum of this expression is $|A_2|$, then Player I gets all $A_2$,
and since the difference between $A_1$ and $A_2$ is only $\Delta_1$, which
player I doesn't really want ($A_1\cap A \subseteq A_2$), player I couldn't do
better by asking $A_1$.

If the minimum of $C(A_2,B)$ is one of the two elements in the maximum
argument, this element is also the minimal value for the demand $A_1$,
because those two elements has the same size in $C(A_2,B)$ and in $C(A_1,B)$, and
the third expression ($|A_1|$) has larger value than $|A_2|$.

Regarding the position of the allocation itself, from non-wastefulness, the
allocation must contain $A_1\cap \bar{B}$ in case the player demanded $A_1$ and
$A_2\cap \bar{B}$ in case the player demanded $A_2$. Those subsets are identical
$A_1\cap \bar{B}=A_2\cap \bar{B}$. In case player I demanded $A_2$, The rest of
the allocation would have to be from  $A\cap A_1 \cap B \subseteq A$ (his real
will). Therefore, also in this case, player I couldn't do
better by asking $A_1$ instead of $A_2$.

In either case, player I can't lose by removing $\Delta_1$ from his demand, and
asking for $A_2$.

We will define $A_3$ as $A_2$, only with $|\Delta_2|=0$ (adding back to the
demand the subset $\Delta_2$ that was removed from $A$) $A_3\equiv
A_2\cup\Delta_2$.
All of the three elements in the expression for $C(A_3,B)$ can only be larger
than in $C(A_2,B)$, therefore $C(A_2,B)\leq C(A_3,B)$. Since the intervals
that aren't joint with $B$ are the same between $A_2$ and $A_3$, and from
non-wastefulness Player I would get all of them, the extra intervals that player
I would get by asking $A_3$, would have to be from $A_3 \cap B$, but since
$|\Delta_1|=0$ there aren't any intervals in $A_3 \cap B$ which Player I doesn't
really want.
Therefore he can only benefit from reducing $|\Delta_2|$ to 0.

We will define $A_4$ as $A_3$, only with $|\Delta_3|=0$: $A_4\equiv
A_3\setminus\Delta_3$.
Each of the three expressions in $C(A_4,B)$ is between
$\theta\cdot|\Delta_3|$ and $|\Delta_3|$, and is smaller than in $C(A_3,B)$.
Therefore by demanding $A_4$ instead of $A_3$, player I would get a smaller piece of cake.
But the piece of cake that the player would get from demanding $A_3$ would
necessarily include the subset $\Delta_3$ (because this subset is not
included in player II demand and the non-wastefulness of $F$). This means that
out of the piece that player I will receive, there is a subset with total size
of $|\Delta_3|$ which he doesn't want. There aren't such intervals in $A_4$,
Therefore the whole piece $C(A_4,B)$ would be from areas that the player
wants, and he can only benefit (between 0 and $(1-\theta)\cdot|\Delta_4|$)
by demanding $A_4$ instead of $A_3$.

For the last stage, we will show that $A_5$, which is defined as $A_4$, only
with $|\Delta_4|=0$, meaning $A_5\equiv A_4\cup\Delta_4$ (adding back the
subset that was removed in $\Delta_4$ to the demand), is better for player I than $A_4$.
$C(A_5,B)=\min\{|A_5|,\max\{|A_5\setminus B|,\theta\cdot|A_5 \cup
B|\}\}=\min\{|A_4|+|\Delta_4|,\max\{|A_4\setminus B|+|\Delta_4|,\theta\cdot(|A_4
\cup B|+|\Delta_4|)\}\}$.
Each one of the three elements in this expression is larger than the elements in
$C(A_4,B)$, therefore $C(A_4,B)\leq C(A_5,B)$.
Since $|\Delta_1|=0$ and $|\Delta_3|=0$ there aren't any intervals in $A_5$ that
player I doesn't really wants, therefore Player I can only benefit from zeroing
$\Delta_4$.

At this point we should notice that $A_5=A$, therefore for any $A$ and $A_1$
player I can only benefit from telling the truth, and mechanism $F$ is IC.
\end{proof}

\section{Reduction From the General to the Aligned Model}
\reductionB*
\begin{proof}

For $a+b \leq 1$, The mechanism $f$ would result in $c(a,b)=a$ and
$d(a,b)=b$ according to its definition.
We can match those $a,b$ the pair $A=[0,a],B=[1-b,1]$.
Since $a+b \leq 1$, those interval are disjoint and from the
non-wastefulness of $F$, $F(A,B)$ would result $C=[0,a],D=[1-b,1]$,
therefore the condition is satisfied.

Otherwise, $a+b > 1$. Since from the non-wastefulness of $F$,
$|\tilde{C}|+|\tilde{D}|=1$, at least one of the following inequalities has to
be fulfilled: $a>|\tilde{C}|,b>|\tilde{D}|$.

In case both inequalities are true, mechanism $f$ would result
$c(a,b)=|\tilde{C}|,d(a,b)=|\tilde{D}|$.
We will now look at mechanism $F$. We know that in the case of two players each
of whom demands the whole interval $[0,1]$, player I would get the piece
$\tilde{C}$ and player II the piece $\tilde{D}$.

Say player II wants the whole cake, and player I only wants a smaller
interval -- the interval $A_1$, which is of size $a$ and contains all of
$\tilde{C}$ (we know that $a>|\tilde{C}|$). Mechanism $F$ would have to
give player I a piece $C_1\subset A_1$ which is of size $|\tilde{C}|$. Otherwise
-- if $|C_1|<|\tilde{C}|$, this player would prefer to lie and demand the whole
cake. If $|C_1|>|\tilde{C}|$, a player I who wants the whole cake would prefer
to lie and ask for $A_1$. In both cases it contradicts $F$'s
incentive-compatibility.

Since $F$ is non-wasteful and player II wants the
whole cake, he would receive everything that player I didn't received. Namely,
the piece $D_1=[0,1]\setminus C_1$, with size of:
$|D_1|=1-|C_1|=1-|\tilde{C}|=|\tilde{D}|$.
We should notice that specifically, since $C_1\subset A_1$, $[0,1]\setminus A_1
\subset D_1$.

Now, say that player I still wants the piece $A_1$ as mentioned, and player II
wants a piece $B_2$ of size $b$ which contains all of $D_1$ (we know that
$b>|\tilde{D}|=|D_1|$). Mechanism $F$ would have to give player II a piece
$D_2\subset B_2$, sized $|\tilde{D}|$ (follows from incentive-compatibility of
$F$, from the same reasons as above). $B_2\setminus A_1 \subset D_2$ (from
non-wastefulness -- player II must receive everything that only he demanded).
We also know that $[0,1]\setminus A_1 \subset D_1\setminus A_1 \subset
B_2\setminus A_1\subset D_2$, Meaning $[0,1]\setminus A_1 = A_1^C
\subset D_2$.
As for the piece $C_2$, which player I would receive after the last change of
player II demand, he would get everything that is in his demand $A_1$ and
wasn't allocated to player II. $|C_2|=|A_1\setminus D_2|=|A_1\cup
D_2|-|D_2|\underset{A_1^C \subset D_2}{=}1-|D_2|=|\tilde{C}|$.

Therefore, If player I demands the set $A_1$ ($|A_1|=a$) and player II
demands the set $B_2$ ($|B_2|=b$), $F$ would give player I the piece $C_2$
of size $|\tilde{C}|$, and player II the piece $D_2$ of size $|\tilde{D}|$.
Those sizes matches mechanism $f(a,b)$.

The last possible case is when only one of the two inequalities is true (without
loss of generality, $a\leq|\tilde{C}|,b>|\tilde{D}|$).
Mechanism $f$, as we defined it above, would result
$c(a,b)=a~,~d(a,b)=1-a$ (since
$b>|\tilde{D}|\rightarrow|\tilde{C}|>1-b$, and $a\leq|\tilde{C}|$).

Again, we will start by examining $F$ in a situation in which both players want
the whole cake. Player I receive the piece $\tilde{C}$ and player II the piece
$\tilde{D}$.

Say an alternate player I only wants a piece $A_1$ of size $a$, such that
$A_1\subseteq\tilde{C}$ (we know that $a\leq|\tilde{C}|$), while player II
still wants the whole cake. As in the previous case, from
incentive-compatibility of $F$, player I should receive the piece $C_1=A_1$.
From non-wastefulness, player II would receive $D_1=[0,1]\setminus A_1$,
$|D_1|=1-a$.

In a situation that player I demands the piece $A_1$ and
player II wants a piece $B_2$, such that $|B_2|=b$ and $D_1=[0,1]\setminus
A_1 \subset B_2$ (possible because $|A_1|=a$ and $b>1-a$). From
incentive-compatibility, $|D_2|=|D_1|=1-a$. Since $|A_1\cup B_2|=1$, from
non-wastefulness $|C_2|=1-|D_2|=a$.

Therefore, If player I demands the set $A_1$ ($|A_1|=a$) and player II
demands the set $B_2$ ($|B_2|=b$), $F$ would give player I the piece $C_2$
of size $a$, and player II the piece $D_2$ of size $(1-a)$. Those sizes
matches mechanism $f(a,b)$.

Note that in all cases where $a+b\geq 1$, The demands of the players $A$ and $B$
were defined such that $[0,1]\setminus A \subseteq B$, meaning $A\cup B=[0,1]$.
\end{proof}
\end{document}